\documentclass{article}
\usepackage[a4paper, margin=1.5in]{geometry}
 
\usepackage{microtype}
\usepackage{url}
\usepackage{algorithm}
\usepackage{algorithmicx}
\usepackage[noend]{algpseudocode}
\usepackage{amsthm,amsmath,amssymb,graphicx,mathrsfs,subcaption}
\usepackage{enumitem}
\usepackage{array}
\usepackage{authblk}
\usepackage{thm-restate} 

\newcommand{\set}[1]{\{#1\}}
\newcommand{\ep}{\epsilon}

\newcommand{\cost}{\text{cost}}

\newtheorem{lemma}{Lemma}

\newtheorem{theorem}{Theorem}

\newcolumntype{M}[1]{>{\centering\arraybackslash}m{#1}}

\title{A PTAS for Bounded-Capacity Vehicle Routing in Planar Graphs}
\author{Amariah Becker\thanks{Research supported by NSF Grant CCF-1409520}}
\affil{Brown University}
\author{Philip Klein\thanks{Research supported by NSF Grant CCF-1409520}}
\affil{Brown University}
\author{Aaron Schild\thanks{Research supported by NSF Grant CCF-1816861}}
\affil{University of California, Berkeley}

\begin{document}

\maketitle

\begin{abstract}
The {\sc Capacitated Vehicle Routing} problem is to find a
minimum-cost set of tours that collectively cover clients in a graph,
such that each tour starts and ends at a specified depot and is
subject to a capacity bound on the number of clients it can serve.  In
this paper, we present a polynomial-time approximation scheme (PTAS)
for instances in which the input graph is planar and the capacity is
bounded.  Previously, only a quasipolynomial-time approximation scheme
was known for these instances.  To obtain this result, we show how to embed planar graphs
into bounded-treewidth graphs while preserving, in expectation, the
client-to-client distances up to a small additive error proportional
to client distances to the depot.
\end{abstract}

\section{Introduction}\label{sec:intro}
We define the {\sc Capacitated Vehicle Routing} problem with capacity $Q>0$ 
as follows.  The input is an undirected graph with nonnegative
edge-lengths, a distinguished vertex $r$ (the \emph{depot}), and a set
$S$ of vertices (the \emph{clients}).  The output is a set of tours
(closed walks), each including the depot, together with an assignment
of clients to tours such that each client belongs to the tour to which
it is assigned, and such that each tour is assigned at most $Q$
clients.  The objective is to minimize the total length of the tours.
We refer to this quantity as the \emph{cost} of the solution.

This problem
arises in both public and commercial settings including
planning school bus routes and package delivery.  {\sc Capacitated
  Vehicle Routing} is NP-hard for any capacity greater than
two~\cite{asano1997}. We provide a polynomial-time approximation
scheme (PTAS) for {\sc Capacitated Vehicle Routing} when the capacity
is bounded and the underlying graph is planar

An \emph{embedding} of a guest graph $G$ in a host graph $H$ is a mapping $\phi: V(G) \longrightarrow V(H)$.  One seeks embeddings in which, for each pair $u,v$ of vertices of $G$, the $u$-to-$v$ distance in $G$ is in some sense approximated by the $\phi(u)$-to-$\phi(v)$ distance in $H$.  One algorithmic strategy for addressing a metric problem is as follows: find an embedding $\phi$ from the input graph $G$ to a graph $H$ with simple structure; find a good solution in $H$; lift the solution to a solution in $G$.  The success of this strategy depends on how easy it is to find a good solution in $H$ and how well distances in $H$ approximate corresponding distances in $G$.
 
In this paper, we give a randomized method for embedding a planar
graph $G$ into a bounded-treewidth host graph $H$ so as to achieve a
certain expected distance approximation guarantee.  There is a
polynomial-time algorithm to find an optimal solution to {\sc
  Capacitated Vehicle Routing} in bounded-treewidth graphs.  This
algorithm is used to find an optimal solution to the problem induced
in $H$.  This solution in the host graph is then lifted to obtain a near-optimal solution in $G$.

\subsection{Related Work}\label{sec:related_work}

\subsubsection*{Capacitated Vehicle Routing} There is a substantial body of work on approximation algorithms for {\sc Capacitated Vehicle Routing}.  As the problem generalizes the {\sc Traveling Salesman Problem} (TSP), for general metrics and values of $Q$, {\sc Capacitated Vehicle Routing} is also APX-hard~\cite{papadimitriou}.  Haimovich and Rinnoy Kan~\cite{haimovich1985} observe the following lower bound.
\begin{equation} \label{eq:lb}
\frac{2}{Q}\sum_{v\in S}{d(v,r)} \leq \cost(OPT)
\end{equation}
where $\cost(OPT)$ denotes the cost of the optimal solution.  They
use this inequality to give a $1+(1-\frac{1}{Q})\alpha$-approximation, where $\alpha$ denotes the approximation ratio of TSP.  Using Christofides 1.5-approximation for TSP~\cite{christofides}, this gives a $2.5-\frac{1}{Q}$ approximation ratio.   For general metrics and values of $Q$ this result has not been substantially improved upon.  Even for tree metrics, the best known approximation ratio for arbitrary values of $Q$ is 4/3, due to Becker~\cite{becker_trees}. While no polynomial-time approximation schemes are known for arbitrary $Q$ for \emph{any} nontrivial metric, recently Becker and Paul~\cite{becker_paul} gave a bicriteria $(1,1+\ep)$ approximation scheme for tree metrics.  It returns a solution of at most the optimal cost, but in which each tour is responsible for at most $(1+\ep)Q$ clients.

One reasonable relaxation is to consider restricted values of $Q$. Even for $Q$ as small as 3, {\sc Capacitated Vehicle Routing} is APX-hard in general metrics~\cite{asano1997}.  On the other hand, for fixed values of $Q$, the problem can be solved in polynomial time on trees and bounded-treewidth graphs.

Much attention has been given to approximation schemes for Euclidean metrics. In the Euclidean plane $\mathbb{R}^2$, PTASs are known for instances in which the value of $Q$ is constant~\cite{haimovich1985}, $O(\log n/\log\log n)$~\cite{asano1997}, and $\Omega(n)$~\cite{asano1997}.  For $\mathbb{R}^3$, a PTAS is known for $Q=O(\log n)$ and for higher dimensions $\mathbb{R}^d$, a PTAS is known for $Q=O(\log^{1/d}n)$\cite{khachay2016}.  For arbitrary values of $Q$, Mathieu and Das designed a quasi-polynomial time approximation scheme (QPTAS) for instances in $\mathbb{R}^2$~\cite{das2010}. No PTAS is known for arbitrary values of $Q$.

There have been a few recent advances in designing approximation
schemes for {\sc Capacitated Vehicle Routing} in non-Euclidean
metrics.  Becker, Klein, and Saulpic \cite{bks_planar} gave a
QPTAS for bounded-capacity instances in planar and bounded-genus graphs.  The same authors gave a PTAS for graphs of bounded highway dimension \cite{bks_hwy_dim}.

\subsubsection*{Metric embeddings} 

There has been much work on metric embeddings. In particular,
Bartal~\cite{Bartal96} gave a randomized algorithm for selecting an
embedding $\phi$ of the input graph into a tree so that, for any
vertices $u$ and $v$ of $G$, the expected $\phi(u)$-to-$\phi(v)$
distance in the tree approximates the $u$-to-$v$ distance in $G$ to
within a polylogarithmic factor.  Fakcharoenphol, Rao, and
Talwar~\cite{FRT04} improved the factor to $O(\log n)$.

Talwar~\cite{talwar2004bypassing} gave a randomized algorithm for
selecting an embedding of a metric space of bounded doubling dimension
and aspect ratio $\Delta$ into a graph whose treewidth is bounded by a
function that is polylogarithmic in $\Delta$; the distances are
approximated to within a factor of $1+\epsilon$.  Feldman, Fung,
K\"onemann, and Post.~\cite{feldmann20151+} built on this result to
obtain a similar embedding theorem for graphs of bounded highway
dimension.

What about planar graphs?  Chakrabarti et al.~\cite{CJLV08} showed a
result that implies that unit-weight planar graphs cannot be embedded into
distributions over $o(\sqrt{n})$-treewidth graphs so as to achieve
approximation to within an $o(\log n)$ factor.

Let us consider distance approximation guarantees with absolute
(rather than relative) error.  Becker, Klein, and
Saulpic~\cite{bks_hwy_dim} gave a deterministic algorithm that, given
a constant $\epsilon>0$, finds an embedding from a graph $G$ of bounded
highway dimension to a bounded-treewith graph $H$ such that, for each
pair $u,v$ of vertices of $G$, the $\phi(u)$-to-$\phi(v)$ distance in
$H$ is at least the $u$-to-$v$ distance in $G$ and exceeds that
distance by at most $\epsilon$ times the $u$-to-$r$ distance
plus the $v$-to-$r$ distance, where $r$ is a given vertex of $G$.
This embedding was used to obtain the previously mentioned PTAS for 
 {\sc Capacitated Vehicle Routing} with bounded capacity on graphs of
 bounded highway dimension.

Recently, Fox-Epstein, Klein, and Schild~\cite{spanner_paper} showed how to embed planar graphs into graphs of bounded treewidth, such that distances are preserved up to a small additive error of $\ep D$, where $D$ is the diameter of the graph.  They show how such an embedding can be used to achieve efficient bicriteria approximation schemes for $k$-{\sc Center} and $d$-{\sc Independent Set}.

\subsection{Main Contributions}\label{sec:contributions}
In this paper we present the first known PTAS for {\sc Capacitated Vehicle Routing} on planar graphs.  We formally state the result as follows.

\begin{theorem}\label{thm:vehicle_routing}
For any $\ep>0$ and capacity $Q$, there is a polynomial-time algorithm
that, given an instance of {\sc Capacitated Vehicle Routing} on planar
graphs with capacity $Q$, returns a solution whose cost is at most $1+\ep$ times optimal.
\end{theorem}

Prior to this work, only a QPTAS was known~\cite{bks_planar} for planar graphs.  As described in Section~\ref{sec:related_work}, PTASs for {\sc Capacitated Vehicle Routing} are known only for very few metrics.  Our result expands this small list to include planar graphs---a graph class that is quite relevant to vehicle-routing problems as many road networks are planar or near-planar.

The basis for our new PTAS is a new metric-embedding theorem.  For a graph
$G$ with edge-lengths and vertices $u$ and $v$, let $d_G(u,v)$ denote the $u$-to-$v$
distance in $G$.

\begin{restatable}{theorem}{embeddingthm}\label{thm:embed}
There is a constant $c$ and a randomized polynomial-time algorithm
that, given a planar graph $G$ with specified root vertex $r$ and
given $0<\ep<1$, computes a graph $H$ with treewidth at most $(\frac{1}{\ep})^{c\ep^{-1}}$ and
an embedding $\phi$ of $G$ into $H$, such that, for every pair of
vertices $u,v$ of $G$, $d_G(u,v) \leq d_H(\phi(u),\phi(v))$ with probability 1, and 
\begin{equation} \label{eq:expected-error}
E[d_H(\phi(u),\phi(v))] \leq d_G(u,v) + \ep[d_G(u,r) +d_G(v,r)]
\end{equation}
\end{restatable}
\noindent The expectation $E[\cdot]$ is over the random choices of the algorithm.

Why does this metric-embedding result give rise to an approximation
scheme for {\sc Capacitated Vehicle Routing}?
We draw on the following observation, which was also used in previous
approximation schemes~\cite{bks_planar,bks_hwy_dim}:
 tours with clients far from the depot can accommodate a larger error.  In particular, each client can be \emph{charged} error that is proportional to its distance to the depot.  In designing an appropriate embedding, we can afford a larger \emph{error allowance} for the clients farther from the depot.

Our new embedding result builds on that of Fox-Epstein et al.~\cite{spanner_paper}.  The challenge in directly applying their embedding result is that it gives an \emph{additive} error bound, proportional to the diameter of the graph.  This error is too large for those clients close to the depot.  Instead, we divide the graph into annuli (\emph{bands}) defined by distance ranges from the depot and apply the embedding result to each induced subgraph independently, with an increasingly large error tolerance for the annuli farthest from the depot.  In this way, each client \emph{can} afford an error proportional to the diameter of the \emph{subgraph} it belongs to.

How can these subgraph embeddings be combined into a global embedding with the desired properties?  In particular, clients that are close to each other in the input graph may be separated into different annuli.  How can we ensure that the embedding approximately preserves these distances while still achieving bounded treewidth?

We show that by randomizing the choice of where to define the annuli boundaries, and connecting all vertices of all subgraph embeddings to a new, global depot,  client distances are approximately preserved (to within their error allowance) \emph{in expectation} by the overall embedding, without substantially increasing the treewidth.  Specifically, we ensure that the annuli are \emph{wide} enough that the probability of nearby clients being separated (and thus generating large error) is small.  Simultaneously, the annuli must be \emph{narrow} enough that, within a given annulus, the clients closest to the depot can afford an error proportional to error allowance of the clients farthest from the depot.

Once the input graph $G$ is embedded in a bounded-treewidth host graph
$H$, a dynamic-programming algorithm can be used to find an optimal
solution to the instance of {\sc Capacitated Vehicle Routing} induced
in $H$, and the solution can be straightforwardly lifted to obtain a
solution in the input graph that in expectation is near-optimal.

Finally we describe how this result can be derandomized by trying all possible (relevant) choices for defining annuli and noting that for \emph{some} such choice, the resulting solution cost must be near-optimal.

\subsection{Outline}\label{sec:outline}
In Section~\ref{sec:prelims} we describe preliminary notation and definitions.  Section~\ref{sec:embedding} describes the details of the embedding and provides an analysis of the desired properties.  In Section~\ref{sec:ptas} we outline our algorithm and prove Theorem~\ref{thm:vehicle_routing}. We conclude with some remarks in Section~\ref{sec:conclusion}.

\section{Preliminaries}\label{sec:prelims}
\subsection{Basics}

Let $G=(V,E)$ denote a graph with vertex set $V$ and edge set $E$.
The graph comes equipped with a 

, and
let $n = |V|$.  As mentioned earlier, for any two vertices $u,v \in
V$, we use $d_G(u,v)$ to denote the $u$-to-$v$ distance in $G$,
i.e. the minimum length of a $u$-to-$v$ path.
We might omit the subscript when the choice of graph is unambiguous.  The \emph{diameter} of a graph $G$ is the maximum distance $d_G(u,v)$ over all choices of $u$ and $v$.
A graph is \emph{planar} if it can be drawn in the plane without any edge crossings.

We use $OPT$ to denote an optimal solution. For a minimization problem, an \emph{$\alpha$-approximation algorithm} is one that returns a solution whose cost is at most $\alpha$ times the cost of $OPT$.  An \emph{approximation scheme} is a family of $(1+\ep)$-approximation algorithms, indexed by $\ep >0$.  A \emph{polynomial-time approximation scheme} (PTAS) is an approximation scheme such that, for each $\ep >0$, the corresponding algorithm runs in $O(n^c)$ time, where $c$ is a constant independent of $n$ but may depend on $\ep$.  A \emph{quasi-polynomial-time approximation scheme} (QPTAS) is an approximation scheme such that, for each $\ep >0$, the corresponding algorithm runs in $O(n^{\log^cn})$ time, where $c$ is a constant independent of $n$ but may depend on $\ep$.  

An \emph{embedding} of a guest graph $G$ into a host graph $H$ is a mapping $\phi:V_G \rightarrow V_H$ of the vertices of $G$ to the vertices of $H$.  

A \emph{tree decomposition} of a graph $G$ is a tree $T$ whose nodes (called \emph{bags}) correspond to subsets of $V$ with the following properties:
\begin{enumerate}
\item For each $v\in V$, $v$ appears in some bag in $T$
\item For each $(u,v) \in E$, $u$ and $v$ appear \emph{together} in some bag in $T$
\item For each $v \in V$, the subtree induced by the bags of $T$ containing $v$ is connected
\end{enumerate}

The \emph{width} of a tree decomposition is the size of the largest bag, and the \emph{treewidth} of a graph $G$ is the minimum width over all tree decompositions of $G$.

\subsection{Problem Statement}

A \emph{tour} in a graph $G$ is a closed path $v_0,v_1,v_2,...,v_L$ such that $v_0 = v_L$ and for all $i \in \{1,2,...,L\}$, $(v_{i-1},v_i)$ is an edge in $G$.  

Given a capacity $Q>0$ and a graph $G = (V,E)$ with specified client set $S\subseteq V$ and depot vertex $r\in V$, the {\sc Capacitated Vehicle Routing} problem is to find a set of tours $\Pi = \{\pi_1,\pi_2,...\pi_{|\Pi|}\}$ that collectively cover all clients and such that each tour includes $r$ and covers at most $Q$ clients.  The cost of a solution is the sum of the tour lengths, and the objective is to minimize this sum.

If a client $s$ is covered by a tour $\pi$, we say that $\pi$ \emph{visits} $s$.  Note that $\pi$ may \emph{pass} many other vertices (including other clients) that it does not cover.

As stated, the problem assumes that each client has unit demand.  In fact, the more general case, where clients have integral demand (assumed to be polynomially bounded) that is allowed to be covered across multiple tours (demand is \emph{divisible}) reduces to the unit-demand case as follows:  For each client $s\in S$ with demand $dem(s) = k$, add $k$ new vertices $\{v_1,v_2,...,v_k\}$ each with unit demand and edges $(s,v_i)$ of length zero, and set $dem(s)$ to zero.  Note that this modification does not affect planarity.  Additionally, since demand is assumed to be polynomially-bounded, the increase in graph size is negligible for the purpose of a PTAS.

For {\sc Capacitated Vehicle Routing} with \emph{indivisible} demands, each client's demand must be covered by a single tour, and a tour can cover at most $Q$ units of client demand.  



We assume values of $\ep$ are less than one.  If not, any $\ep \geq 1$ can be replaced with a number $\ep'$ slightly less than one.  This only helps the approximation guarantee and does not significantly increase runtime.  Of course for very large values of $\ep$, an efficient constant-factor approximation can be used instead (see Section~\ref{sec:related_work}).

\section{Embedding}\label{sec:embedding}

In this section, we prove Theorem~\ref{thm:embed}, which we restate
for convenience:

\addtocounter{theorem}{-1}
\embeddingthm
%
%
The proof uses as a black box the following result from~\cite{spanner_paper}:

\begin{lemma}[\cite{spanner_paper}]\label{lem:spanner}
  There is a number $c$ and a polynomial-time algorithm that, given
  a planar graph $G$ with specified root vertex $r$ and diameter $D$,
  computes a graph $H$ of treewidth at most $(\frac{1}{\ep})^c$ 
  and an embedding $\phi$ of $G$ into $H$ such that, for all vertices $u$ and $v$,
  $$d_G(u,v)\leq d_H(\phi(u),\phi(v)) \leq d_G(u,v) + \ep D$$
\end{lemma}

For notational convenience, instead of
Inequality~\ref{eq:expected-error} of Theorem~\ref{thm:embed}, we prove
\begin{equation} \label{eq:expected-error3}
E[d_H(\phi(u),\phi(v))] \leq d_G(u,v) + 3\ep[d_G(u,r) +
d_G(v,r)]
\end{equation}
from which Theorem~\ref{thm:embed} can be proved by taking $\epsilon'=\epsilon/3$.

Our embedding partitions vertices of $G$ into \emph{bands} of vertices
defined by distances from $r$.  Choose $x \in [0,1]$ uniformly at
random.  Let $B_0$ be the set of vertices $v$ such that
$d_G(r,v) < {\frac{1}{\ep}}^{(x)\frac{1}{\ep}}$, and for
$i \in \{1,2,3,...\}$ let $B_{i}$ be the set of vertices $v$ such that
${\frac{1}{\ep}}^{(i+x-1)\frac{1}{\ep}} \leq d_G(r,v) <
{\frac{1}{\ep}}^{(i+x)\frac{1}{\ep}}$ (see Figure~\ref{fig:bands}).
Let $G_i$ be the subgraph induced by $B_{i}$, together with all
$u$-to-$v$ and $v$-to-$r$ shortest paths for all $u,v\in B_{i}$.  Note
that although the $B_i$ partition $V$, the $G_i$ do not partition
$G$. Note also that the diameter of $G_i$ is at most
$4{\frac{1}{\ep}}^{(i+x)\frac{1}{\ep}}$.  This takes into account the
paths included in $G_i$ that pass through vertices not in $B_i$.

For each $G_i$, let $\phi_i$ be the embedding and let $H_i$ be
the host graph resulting from applying Lemma~\ref{lem:spanner} and
using $\ep' = \ep^{\frac{1}{\ep}+1}$.  Finally, let $H$ be the graph
resulting from adding a new vertex $r'$ and for all $i$ and all
$v \in B_i$ adding an edge $(\phi_i(v),r')$ of length $d_G(v,r)$. Set
$\phi(v)=\phi_i(v)$ for all $v \in B_i-\set{r}$ and set $\phi(r)=r'$. See
Figure~\ref{fig:embedding}.

\begin{figure}[H]     
\centering
\includegraphics[width=0.5\textwidth]{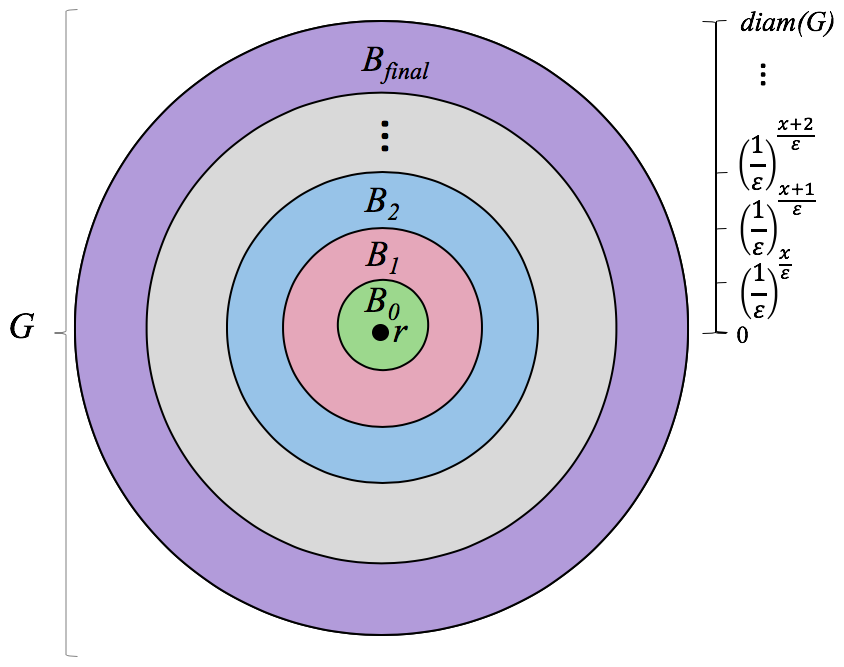}         
\caption{$G$ is divided into bands $B_0,B_1,...,B_{final}$ based on distance from $r$.}         
\label{fig:bands}
\end{figure}     

\begin{figure}[H]     
\centering
\includegraphics[width=0.8\textwidth]{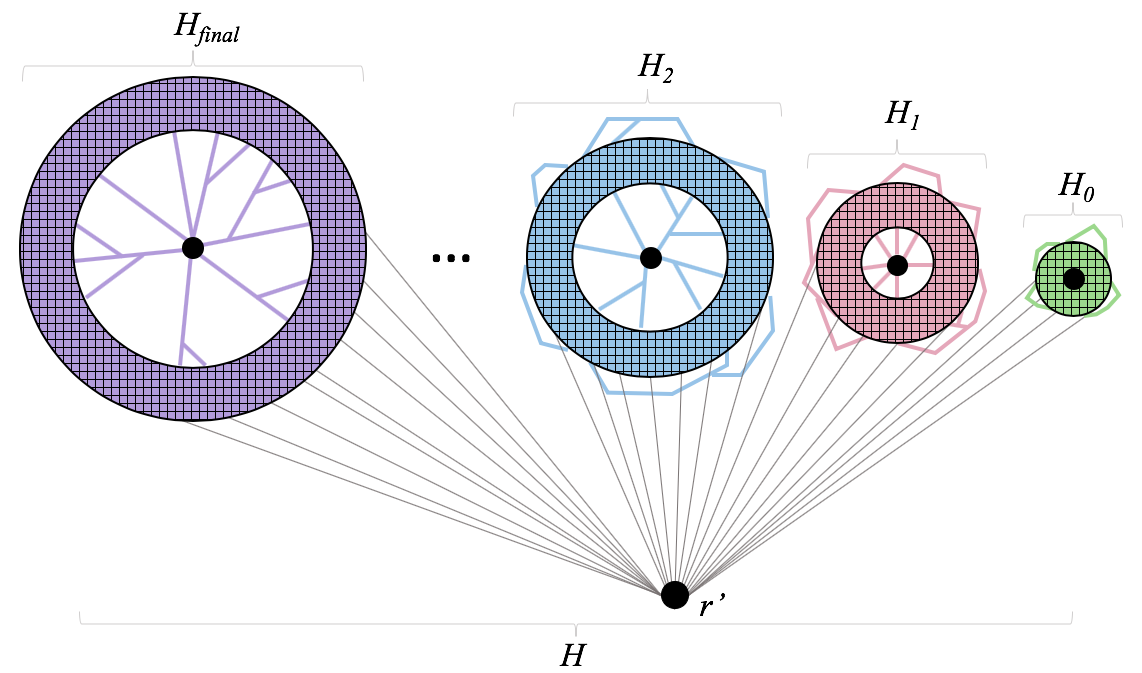}         
\caption{Each subgraph $G_i$ of $G$ is embedded into a host graph $H_i$.
  These graphs are joined via edges to a new depot $r'$ to form a
   host graph for $G$.}         
\label{fig:embedding}
\end{figure}

 Let $H^-$ be the graph obtained from $H$ by deleting $r'$.  The connected components of $H^-$ are $\{H_i\}_i$.  By Lemma~\ref{lem:spanner}, the treewidth of each host graph $H_i$ is at most $(\frac{1}{\ep'})^{c_0}$ = $(\frac{1}{\ep})^{c_0(\ep^{-1} + 1)}$ for some constant $c_0$. This also bounds the treewidth of $H^-$.  Adding a single vertex to a graph increases the treewidth by at most one, so after adding $r'$ back, the treewidth of $H$ is $(\frac{1}{\ep})^{c_0(\ep^{-1} + 1)}+1 = (\frac{1}{\ep})^{c_1\ep^{-1}}$ for some constant $c_1$.

As for the metric approximation, it is clear that $d_G(u,v) \leq
d_H(\phi(u),\phi(v))$ with probability 1.  
We use the following lemma to prove Equation~\ref{eq:expected-error3}.

\begin{lemma}\label{lem:prob}
If $\ep d_G(v,r) < d_G(u,r) \leq d_G(v,r)$, then the probability that $u$ and $v$ are in different bands is at most $\ep$.
\end{lemma}
\begin{proof}
Let $i$ be the nonnegative integer such that $d_G(u,r) =
{\frac{1}{\ep}}^{(i+a)\frac{1}{\ep}}$ for some $a \in [0,1]$.  Let $b$
be the number such that $d_G(v,r) = {\frac{1}{\ep}}^{(i+b)\frac{1}{\ep}}$.

$$\frac{1}{\epsilon} \geq \frac{d_G(v,r)}{d_G(u,r)} =
\frac{{\frac{1}{\ep}}^{(i+b)\frac{1}{\ep}}}{{\frac{1}{\ep}}^{(i+a)\frac{1}{\ep}}}
= {\frac{1}{\ep}}^{(b-a)\frac{1}{\ep}}$$
Therefore
$$b-a \leq \ep$$

Consider two cases.  If $b\leq 1$, then the probability that $u$ and $v$ are in different bands is $Pr[a\leq x < b] \leq \ep$.  

If $b > 1$ then the probability that $u$ and $v$ are in different bands is $Pr[x \geq a \text{ or } x \leq b-1] \leq 1-a + b-1 = b-a \leq \ep$
\end{proof}

We now prove Equation~\ref{eq:expected-error3}.  Let $u$ and $v$ be
vertices in $G$. Without loss of generality, assume $d_G(u,r) \leq d_G(v,r)$.  First we address the case where $d_G(u,r) \leq \ep d_G(v,r)$.  Since $\phi(u)$ and $\phi(v)$ are both adjacent to $r'$ in $H$, $d_H(\phi(u),\phi(v)) \leq d_H(\phi(u),r') + d_H(\phi(v),r') = d_G(u,r) + d_G(v,r) \leq 2d_G(u,r) + d_G(u,v) \leq d_G(u,v) + 2\ep d_G(v,r)$.  Therefore $E[d_H(\phi(u),\phi(v))] \leq d_G(u,v) + 3\ep[d_G(u,r) + d_G(v,r)]$

Now, suppose $d_G(u,r) > \ep d_G(v,r)$.  If $u$ and $v$ are in the
same band $B_i$, then by Lemma~\ref{lem:spanner},
\begin{eqnarray*}
d_{H}(\phi(u),\phi(v))&\leq& d_{H_i}(\phi(u),\phi(v)) \leq d_G(u,v) +
                             \ep' \frac{1}{\epsilon}^{(i+x)\frac{1}{\epsilon}}   \\
&\leq& d_G(u,v) + \ep^{\frac{1}{\ep}+1}
4{\frac{1}{\ep}}^{(i+x)\frac{1}{\ep}}\\
& = & d_G(u,v) +
\ep4{\frac{1}{\ep}}^{(i+x-1)\frac{1}{\ep}} \leq d_G(u,v) +
       2\ep(d_G(u,r)+d_G(v,r))
\end{eqnarray*}
In the final inequality, when
$i=0$, we use the fact that all nonzero distances are at least one to
give a lower bound on $d_G(u,r)$ and $d_G(v,r)$.

If $u$ and $v$ are in different bands, then since  $\phi(u)$ and $\phi(v)$ are both adjacent to $r'$ in $H$, $d_H(\phi(u),\phi(v)) \leq d_H(\phi(u),r') + d_H(\phi(v),r') = d_G(u,r) + d_G(v,r)$.  By Lemma~\ref{lem:prob}, this case occurs with probability at most $\ep$.  

Therefore $E[d_H(\phi(u),\phi(v))] \leq (d_G(u,v) +
2\ep(d_G(u,r)+d_G(v,r))) + \ep[d_G(u,r) + d_G(v,r)] \leq d_G(u,v) +
3\ep[d_G(u,r) + d_G(v,r)]$, which proves Inequality~\ref{eq:expected-error3}.

The construction does not depend on planarity only via Lemma~\ref{lem:spanner}.
For the sake of future uses of the construction
with other graph classes, we state a lemma.

\begin{lemma}\label{lem:reduction}
Let $\mathcal F$ be a family of graphs closed under
vertex-induced subgraphs.  Suppose that there is a function $f$ and a
polynomial-time algorithm that, for any graph $G$ in $\mathcal F$,
computes a graph $H$ of treewidth at most $f(\epsilon)$ and an
embedding $\phi$ of $G$ into $H$ such that, for all vertices $u$ and $v$,
  $$d_G(u,v)\leq d_H(\phi(u),\phi(v)) \leq d_G(u,v) + \ep D$$ 
Then there is a function $g$ and a randomized polynomial-time algorithm that, for any graph $G$ in $\mathcal F$,
computes a graph $H$ with treewidth at most $g(\epsilon)$ and
an embedding $\phi$ of $G$ into $H$, such that, for every pair of
vertices $u,v$ of $G$, with probability 1 $d_G(u,v) \leq
d_H(\phi(u),\phi(v))$, and
$$E[d_H(\phi(u),\phi(v))] \leq d_G(u,v) + \ep\left[(d_G(u,r) +d_G(v,r)\right]$$
\end{lemma}

\section{PTAS for Capacitated Vehicle Routing}\label{sec:ptas}

In this section, we show how to use the embedding of Section~\ref{sec:embedding} to give a PTAS for {\sc Capacitated Vehicle Routing}, proving Theorem~\ref{thm:vehicle_routing}.

\subsection{Randomized algorithm}\label{sec:relaxed}

We first prove a slight relaxation of
Theorem~\ref{thm:vehicle_routing} in which the algorithm is
randomized, and the solution value is near-optimal \emph{in expectation}.  We then show in Section~\ref{sec:derand} how to derandomize the result.

\begin{theorem}\label{thm:expectation_PTAS}
For any $\ep>0$ and capacity $Q$, there is a randomized algorithm for {\sc Capacitated Vehicle Routing} on planar graphs that in polynomial time returns a solution whose expected value is at most $1+\ep$ times optimal.
\end{theorem}
Our result depends on the following lemma, which is proved in~\cite{BeckerKS17}, the full version of \cite{bks_hwy_dim}.  

\begin{lemma}[Lemma 20 in \cite{bks_hwy_dim}, Lemma 15 in \cite{BeckerKS17}]\label{lem:dp}
Given an instance of {\sc Capacitated Vehicle Routing} with capacity $Q$ on a graph $G$ with treewidth $w$, there is a dynamic-programming algorithm that finds an optimal solution in $n^{O(wQ)}$ time.
\end{lemma}

Given the dynamic program of Lemma~\ref{lem:dp} and the embedding of Theorem~\ref{thm:embed} as black boxes, the algorithm is as follows.  First, the graph $G$ is embedded as in Theorem~\ref{thm:embed} using $\hat{\ep}=\ep/3Q$ into a host graph $H$ with treewidth $(\frac{1}{\hat{\ep}})^{c\hat{\ep}^{-1}}$ for some constant $c$, and $d_G(u,v)\leq E[d_H(\phi(u),\phi(v))] \leq d_G(u,v) + 3\hat{\ep}(d_G(u,r) + d_G(v,r))$ for all vertices $u$ and $v$.  The dynamic program of Lemma~\ref{lem:dp} is then applied to $H$. The resulting solution $SOL_H$ in $H$ is then mapped back to a solution $SOL_G$ in $G$ which is returned by the algorithm.

Note that the tours in any vehicle-routing solution can be defined by specifying the order in which clients are visited.  In particular, we use $(u,v)\in SOL$ to denote that $u$ and $v$ are consecutive clients visited by the solution, noting that $u$ or $v$ may actually be the depot.  In this way, a solution in $H$ is easily mapped back to a corresponding solution in $G$, as $(u,v)\in SOL_G$ if and only if $(\phi(u),\phi(v))\in SOL_H$.

\medskip
We now prove Theorem~\ref{thm:expectation_PTAS} by analyzing this algorithm.

\begin{lemma}\label{lem:performance}
For any $\ep>0$, the algorithm described above finds a solution whose expected value is at most $1+\ep$ times optimal.
\end{lemma}
\begin{proof}

Let $OPT$ be the optimal solution in $G$ and let $OPT_H$ be the corresponding induced solution in $H$.  Since the dynamic program finds an optimal solution in $H$, we have $\cost_H(SOL_H)\leq \cost_H(OPT_H)$.  Additionally, since distances in $H$ are no shorter than distances in $G$, $\cost_G(SOL_G) \leq \cost_H(SOL_H)$.  Putting these pieces together, we have

\begin{eqnarray*}
  E[\cost_G(SOL_G)] &\leq& E[\cost_H(SOL_H)]\\
                    &\leq& E[\cost_H(OPT_H)]\\
                                      & =& E[\sum_{(u,v)\in OPT}d_H(\phi(u),\phi(v))]\\
                                      &=& \sum_{(u,v)\in OPT}E[d_H(\phi(u),\phi(v))] \\
  &\leq &  \sum_{(u,v)\in OPT}d_G(u,v) + 3\hat{\ep}(d_G(u,r) +  d_G(v,r))\\
                                      &=& \sum_{(u,v)\in OPT}d_G(u,v) + 6\hat{\ep}\sum_{v\in S}d_G(v,r)\\
                                      &\leq &\cost_G(OPT) + 6\hat{\ep}\frac{Q}{2}\cost_G(OPT)\\
                                      &=& (1+\ep)\cost_G(OPT)
\end{eqnarray*}
where the final inequality comes from Inequality~\ref{eq:lb} (see Section~\ref{sec:related_work}).
\end{proof}

The following lemma completes the proof of Theorem~\ref{thm:expectation_PTAS}.

\begin{lemma}\label{lem:runtime}
For any $Q,\ep>0$, the algorithm described above runs in polynomial time.
\end{lemma}
\begin{proof}
By Lemma~\ref{lem:spanner}, computing $H$ and the embedding of $G$ into $H$ takes polynomial time.
By Lemma~\ref{lem:dp}, the dynamic program runs in $|V_H|^{O(wQ)}$ time, where $w$ is the treewidth of $H$.  By Theorem~\ref{thm:embed}, $w = (\frac{1}{\hat{\ep}})^{c\hat{\ep}^{-1}} = (\frac{Q}{\ep})^{c'Q\ep^{-1}}$, where $c$ and $c'$ are constants independent of $|V_H|$.  

The algorithm therefore runs in $|V_H|^{(Q\ep^{-1})^{O(Q\ep^{-1})}}$
time.  Finally, since $|V_H|$ is polynomial in the size of $G$,  for
fixed $Q$ and $\ep$, the running time is polynomial.
\end{proof}

\subsection{Derandomization}\label{sec:derand}

The algorithm can be derandomized using a standard technique.  The
embedding of Theorem~\ref{thm:embed} partitions the vertices of the
input graph into rings depending on a value $x$ chosen
uniformly at random from $[0,1]$.  However, the partition depends on
the distances of vertices from the root $r$.  It follows that the
number of partitions that can arise from different choices of $x$ is
at most the number of vertices.  The deterministic algorithm tries
each of these partitions, finding the corresponding solution, and
returns the least costly of these solutions.

In particular, consider the optimum solution $OPT$.  As shown in Section~\ref{sec:relaxed},  

$$E[\sum_{(u,v) \in OPT}d_H(\phi(u),\phi(v))] $$ $$= \sum_{(u,v) \in OPT}E[d_H(\phi(u),\phi(v))] $$ $$\leq (1+\ep)\cost_G(OPT)$$.

Therefore, for some choice of $x$, the induced cost of $OPT$ in $H$ is nearly
optimal, and the dynamic program will find a solution that costs at
most as much.
This completes the proof of Theorem~\ref{thm:vehicle_routing}.

\section{Conclusion}\label{sec:conclusion}

In this paper, we present the first PTAS for {\sc Capacitated Vehicle
  Routing} in planar graphs.  Although the approximation scheme takes polynomial
time, it is not an \emph{efficient} PTAS (one whose running time is
bounded by a polynomial whose degree is independent of the value of
$\epsilon$).  It is an open question as to whether an efficient PTAS
exists.  It is also open whether a PTAS exists when the
capacity $Q$ is unbounded.

\bibliography{main}

\end{document}